\documentclass[a4paper]{article}
\usepackage{spconf,cite}
\ninept
\usepackage[utf8]{inputenc}
\usepackage[T1]{fontenc}
\usepackage[draft]{hyperref}

\usepackage{setspace}

\usepackage{amsmath,amssymb,amsthm}
\usepackage{thmtools}
\DeclareMathOperator{\diag}{diag}

\theoremstyle{plain}
\newtheorem{theorem}{Theorem}
\newtheorem{lemma}{Lemma}

\newtheorem{definition}{Definition}

\theoremstyle{remark}

\newcommand*{\Scale}[2][4]{\scalebox{#1}{\ensuremath{#2}}}

\usepackage{graphicx}
\graphicspath{{figures/}}

\usepackage[font=footnotesize]{caption}
\usepackage[labelformat=simple,font=footnotesize]{subcaption}
\title{Localization Bounds for the Graph Translation}
\name{Benjamin Girault$^{\star}$, Paulo Gonçalves$^{\diamond}$,
Shrikanth S. Narayanan$^{\star}$, Antonio Ortega$^{\star}$\thanks{This work was supported in part 
by NSF under grants CCF-1410009, CCF-1527874, CCF-1029373.}}
\address{$^\star$ University of Southern California, Los Angeles, CA 90089, USA \\
$^{\diamond}$ Univ Lyon, Inria, ENS de Lyon, CNRS, UCB Lyon 1, 69342, Lyon, FRANCE}

\begin{document}
\setlength{\abovedisplayskip}{7pt}
\setlength{\belowdisplayskip}{7pt}
\maketitle
\begin{abstract}
  The graph translation operator has been defined with good spectral properties in mind, and in 
  particular with the end goal of being an isometric operator. Unfortunately, the resulting 
  definitions do not provide good intuitions on a vertex-domain interpretation. In this 
  paper, we show that this operator does have a vertex-domain interpretation as a diffusion 
  operator using a polynomial approximation. 
  We show that its impulse response exhibit an exponential decay of the energy way from the 
  impulse, demonstrating localization preservation. 
  Additionally, we formalize several techniques that can be used to study 
  other graph signal operators.
\end{abstract}
\section{Introduction}
\label{sec:intro}

The field of graph signal processing aims at extending the tools of classical 
signal processing to irregular domains, and more precisely to finite and discrete irregular 
structures. In the 
past few years, the field has seen many successes, such as the graph Fourier transform 
\cite{Shuman.IEEESP.2013,Sandryhaila.TSP.2013}, shift invariant filters 
\cite{Sandryhaila.TSP.2013}, ARMA models \cite{Loukas.IEEESP.2015}, several graph wavelet 
transforms \cite{Coifman.ACHA.2006.Wave,Hammond.ACHA.2011,Leonardi.TSP.2013,Shuman.TSP.2015}, or 
vertex-frequency decompositions \cite{Shuman.ACHA.2016,Shuman.TSP.2015}. This is by no means an 
exhaustive list of GSP tools introduced so far, but an illustration to how tools involving the time 
shift operator have been extended to the graph setting. In particular, several 
operators equivalent to the time shift have been devised for graph domains. Chronologically, 
the first of these is the \textit{graph shift} \cite{Sandryhaila.TSP.2013} defined in the vertex 
domain as an operator \textit{shifting} the energy from one vertex to its neighbors and according 
to 
the edges weights. The second is the \textit{generalized translation} \cite{Shuman.ACHA.2016} 
defined as a generalized convolution by a delta signal and having the property of localizing smooth 
signals about said delta signal. The operator we are interested in is the \textit{graph 
translation} \cite{Girault.SPL.2015} defined in the Fourier domain as an isometric operator 
shifting the Fourier components by a complex phase and according to the graph frequencies. 
Finally, in \cite{Gavili.ARXIV.2015}, the authors identified the need of an isometric 
operator after \cite{Girault.SPL.2015}, but the phase shifts of the Fourier components are not 
related to the graph frequencies. Of these shift operators only the last two are isometric 
\cite{Girault.SPL.2015}, \textit{i.e.} the only ones preserving the energy of a signal, and only 
the 
graph translation has meaningful (\textit{w.r.t.} the graph) complex phase shifts.

Numerical experiments with this operator have shown a localization preservation property 
\cite{Girault.SPL.2015,Girault.THESIS.2015}, which has not been shown analytically. The 
localization property can be formalized in three different ways from the impulse response. The 
strictest is to 
verify that the impulse response is itself an impulse, as verified by the time shift operator. We 
showed in \cite{Girault.THESIS.2015} that it is impossible to have both the isometry and this 
strict property. The second definition is to have an impulse response compactly supported. This 
property is shown by the graph shift. 
We believe that this definition is also incompatible with isometry, but this has not yet been 
proven. 
We are focusing here on the third definition involving an 
exponential decay of the energy as vertices get farther away from the center of the impulse. 
The localization property of \cite{Shuman.ACHA.2016} is of this kind. We use a similar method and 
draw a general framework for showing such a property. More precisely, given an operator, we 
approximate it using a polynomial of the adjacency matrix. Such a polynomial operator has itself an 
impulse response compactly supported since the $k^\text{th}$ power of the adjacency matrix 
corresponds to the $k$-hop paths in the graph. Showing that the error made by the approximation 
decays exponentially as the order of the polynomial increases is then enough to show that the 
operator is preserving localization according to the third definition. We show that the graph 
translation verifies this third definition of localization property.

We start by setting the definitions and notations of GSP we use. Then we give several general 
results and techniques that can be used to bound the response of any operator written as a function 
of a given matrix (typically the Laplacian matrix or the adjacency matrix). Finally, we recall the 
definition of the graph translation, and give its associated localization preservation result.

\section{Graph Signal Processing}
\label{sec:background}

Let $\mathcal{G}=(V,E,w)$ be a weighted undirected graph, with $V$ the set of vertices with 
$|V|=N$, $E$ the set of edges such that if $ij$ is an edge, $ji$ is also an edge, and 
$w:E\rightarrow \mathbb{R}$ the weight of the edges. We denote $\mathbf{A}$ the adjacency matrix 
with $A_{ij}=w(i,j)$ the weight of the corresponding edge, $A_{ij}=0$ if no such edge exists. Let 
$\mathbf{D}=\diag(d_1,\dots,d_N)$ be the diagonal matrix of vertex degrees $d_i=\sum_j A_{ij}$. Let 
$\mathbf{L}=\mathbf{D}-\mathbf{A}$ be the Laplacian matrix of the graph, and 
$\mathcal{L}=\mathbf{D}^{-1/2}\mathbf{L}\mathbf{D}^{-1/2}$ the normalized Laplacian 
matrix.

The three matrices $\mathbf{A}$, $\mathbf{L}$ and $\mathcal{L}$ are real symmetric, and as such 
diagonalizable in an orthonormal basis. In particular, we have 
$\mathbf{L}=\mathbf{U}\mathbf{\Lambda} \mathbf{U}^*$ with $\mathbf{\Lambda}$ a diagonal matrix, and 
$\mathbf{U}$ the unitary\footnote{$U$ being real, it is orthogonal, but since we are dealing 
with complex operators and signal, the unitary property is more convenient.} eigenvector matrix. 
The Graph Fourier Transform (GFT) is then 
defined as the projection on those eigenvectors and denoted 
$\widehat{\mathbf{x}}=\mathbf{U}^*\mathbf{x}$. Similar approaches can be derived using $\mathbf{A}$ 
and $\mathcal{L}$. The orthonormality allows then for Parseval's identity to be verified: 
$\sum_i|x_i|^2=\|\mathbf{x}\|_2^2=\|\widehat{\mathbf{x}}\|_2^2=\sum_l|\widehat{x}(l)|^2$, where 
the signal energy in both the vertex and spectral domains is given by the $\ell_2$-norm.

Let $\rho_\mathcal{G}^2=\max_i 2d_i(d_i+\bar{d_i})$, with $\bar{d_i}=\sum_j A_{ij}d_j/d_i$. 
The reduced\footnote{In the temporal domain, reduced frequencies lie in the interval $[-1/2,1/2]$, 
with low frequencies close to zero.} graph frequencies associated to the Laplacian-based GFT are 
then $\nu_l=\frac{1}{2}\sqrt{\lambda_l/\rho_\mathcal{G}}\in[0,\frac{1}{2}]$, where $\lambda_l$ is 
the $l^\text{th}$ eigenvalue\footnote{The rationale behind the square root is the fact that the 
eigenvalues of the continuous Laplacian are squared frequencies (hence the loss of the sign), and 
the discrete graph Laplacian is similar to the continuous Laplacian \cite{Shuman.IEEESP.2013}.} of 
$\mathbf{L}$. Similarly, for the GFT based on $\mathcal{L}$, we have 
$\nu_l=\frac{1}{2}\sqrt{\mu_l/2}\in[0,\frac{1}{2}]$, and for $\mathbf{A}$, we have 
$\nu_l=1-\gamma_l/\gamma_\text{max}\in[0,2]$ (with $\mu_l$ and $\gamma_l$ the eigenvalues of 
$\mathcal{L}$ and $A$ respectively). We refer the interested reader to 
\cite{Girault.SPL.2015,Girault.THESIS.2015,Sandryhaila.TSP.2014} for justifications of these 
definitions.

Finally, we call a convolutive operator a graph filter $\mathbf{H}$ such that there exists a graph 
signal $\mathbf{h}$ verifying $\mathbf{Hx}=\mathbf{h}*\mathbf{x}$ where $*$ is the convolution 
operation defined as $\widehat{\mathbf{h}*\mathbf{x}}(l)=\widehat{h}(l)\widehat{x}(l)$. We denote 
$\widehat{\mathbf{H}}$ the dual operator in the spectral domain such that 
$\widehat{\mathbf{Hx}}=\widehat{\mathbf{H}}\widehat{\mathbf{x}}$. Note that equivalently, $\mathbf 
H$ and $\mathbf L$ (resp. $\mathcal{L}$, $\mathbf A$) are jointly diagonalizable or 
$\widehat{\mathbf H}$ is a diagonal matrix.

\section{Preliminary: Polynomial Approximation Bounds}

We state the main fundamental result we rely on:

\begin{theorem}\label{thm:conv_op_approx_bound}
  Let $f(x)$ be an analytical function and $p_K^{(f)}(x)$ a polynomial approximation of degree $K$ 
  of $f$ such that $|f(x)-p_K^{(f)}(x)|\leq\kappa_f(K)$ on $X\subseteq\mathbb{R}$. Let $\mathbf{M}$ 
  be an Hermitian\footnote{$\mathbf{M}$ and its conjugate transpose $\mathbf{M}^*$ are equal.} 
  matrix with eigenvalues in $X$. We have then:
  \[
    \left\|f(\mathbf{M})\mathbf{x}-p_K^{(f)}(\mathbf{M})\mathbf{x}\right\|_2
    \leq \kappa_f(K) \|\mathbf{x}\|_2\text{.}
  \]
\end{theorem}
\begin{proof}
  First of all, $\mathbf M$ being Hermitian, it can be diagonalizable into $\mathbf M=\mathbf 
  V\mathbf\Theta \mathbf V^*$, with $\mathbf\Theta$ diagonal and $\mathbf V$ unitary. It follows 
  that $f(\mathbf M)=\mathbf Vf(\mathbf\Theta)\mathbf V^*$ and $p_K(\mathbf 
  M)=\mathbf Vp_K(\mathbf\Theta)\mathbf V^*$. We 
  have then:
  \begin{align*}
    \left\|f(\mathbf{M})\mathbf{x}-p_K(\mathbf{M})\mathbf{x}\right\|_2^2
    &=\left\|\left[f(\mathbf\Theta)-p_K(\mathbf\Theta)\right]\mathbf V^*\mathbf{x}\right\|_2^2 \\
    &=\sum_l |f(\theta_l)-p_K(\theta_l)|^2|\mathbf{\widehat x}(l)|^2 \nonumber \\
    &\leq \kappa_f(K)^2\sum_l |\mathbf{\widehat x}(l)|^2 \\
    &= \kappa_f(K)^2\|\mathbf{\widehat x}\|_2^2=\kappa_f(K)^2\|\mathbf{x}\|_2^2\text{.}
  \end{align*}
  where the first and last equality follow from $\mathbf V$ being unitary, and the inequality 
  follows from $\theta_l\in X$.
\end{proof}

We now state three lemmas giving the bound $\kappa_f(K)$.

\begin{lemma}\label{lem:analyt_fun}
  Let $f$ be an analytical function and $p_K$ its polynomial approximation verifying:
  \begin{align*}
    f(x)   &= \sum_{k=0}^\infty f_k\;(x-a)^k &
    p_K(x) &= \sum_{k=0}^K      f_k\;(x-a)^k\text{,}
  \end{align*}
  such that $f$ is well defined\footnote{\textit{i.e.} the sum converges.} on 
  the convex set $X\subseteq\mathbb{R}$. We have then:
  \begin{equation}\label{eq:kapp_taylor}
    \kappa_f(K)=
    \frac{1}{(K+1)!}
    \max_{X} \left\{|f^{(K+1)}|\right\}
    \max_{x\in X}  \left\{|x-a|^{K+1}\right\}
    \text{.}
  \end{equation}
\end{lemma}
\begin{proof}
  We first remark that using the Taylor series decomposition about $a$, we have 
  $f_k=f^{(k)}(a)/k!$. Using Taylor's theorem, for all $x\in X$, there exits 
  $\zeta\in[a,x]$ (or $\zeta\in[x,a]$ if $x<a$) such that:
  \[
    f(x)-p_K(x)=\frac{f^{(K+1)}(\zeta)}{(K+1)!}(x-a)^{K+1}\text{.}
  \]
  Maximizing the absolute value of the quantity above on the set $X$ yields \eqref{eq:kapp_taylor}.
\end{proof}

\begin{lemma}\label{lem:alternating_sum}
  Let $f$ be an analytical function verifying:
  \begin{align*}
    f(x)   &= \sum_{k=0}^\infty (-1)^k f_k\;(x-a)^k 
  \end{align*}
  with $f_k(x-a)^k$ of constant sign for any $x\in X\subseteq\mathbb{R}$. We have then:
  \begin{equation}
    \kappa_f(K)=\left|f_{K+1}(x-a)^{K+1}\right|
    \text{.}
  \end{equation}
\end{lemma}
\begin{proof}
  The sum in $f$ is an alternating series.
\end{proof}

Note that we need $X\subseteq[a,\infty)$ or $X\subseteq(\infty,a]$ for 
\autoref{lem:alternating_sum}.

\begin{lemma}\label{lem:graph_op_comp}
  Let $f(x)=g(x)h(x)$ be the product of two analytical functions well defined on 
  $X\subseteq\mathbb{R}$, and $p^{(f)}_{P,Q}(x)=p^{(g)}_Q(x)p^{(h)}_P(x)$ be its polynomial 
  approximation. Let $\kappa_g(Q)$ and $\kappa_h(P)$ be the associated upper bounds. We have then:
  \[
    \kappa_f(P,Q)=\kappa_g(Q) \max_X|h|+\kappa_h(P) (\max_X|g|+\kappa_g(Q))\text{.}
  \]
\end{lemma}
\begin{proof}
  We split the difference in two parts (we removed the arguments of the functions for conciseness):
  \begin{align*}
    \left|f-p^{(f)}\right|
    &\leq \left|gh-p^{(g)}_Qh\right| + \left|p^{(g)}_Qh-p^{(g)}_Qp^{(h)}_P\right| \\
    &= \left|g-p^{(g)}_Q\right||h|+\left|p^{(g)}_Q\right|\left|h-p^{(h)}_P\right| \\
    &\leq \kappa_g(Q)\max_X|h|+(\max_X |g|+\kappa_g(Q))\kappa_h(P)\text{,}
  \end{align*}
  where we used the property $|p^{(g)}_Q|\leq|g|+|g-p^{(g)}_Q|$ in the second inequality.
\end{proof}

We will be using \autoref{thm:conv_op_approx_bound} with $\mathbf M=\mathbf L/\rho_\mathcal{G}$, 
$\mathbf M=\mathcal{L}/2$ and $\mathbf M=\mathbf I-\mathbf A/\gamma_\text{max}$ in 
\autoref{sec:localization}.

\section{Graph Translation}

The goal of this paper is to show that the graph translation of \cite{Girault.SPL.2015} 
preserves localization in the third sense: The impulse response has a decaying energy from the 
impulse vertex to its neighbors and beyond. 
Given a GFT and a set of graph frequencies $\{\nu_l\}_l$, we define the \textit{Graph Translation} 
operator as the isometric convolutive operator verifying $\widehat{\mathbf T}=\exp\left(-\imath 
2\pi\diag(\nu_0,\cdots,\nu_{N-1})\right)$ \cite{Girault.THESIS.2015}. Isometry is understood here 
as the property $\|\mathbf Tx\|_2=\|x\|_2$. This leads to the following 
equivalent definitions:

\begin{definition}[Graph Translation with $L$]\label{def:trans_l}
  \begin{equation}
    \mathbf T_\mathcal{G} = \exp\left(-\imath \pi 
    \sqrt{\frac{\mathbf L}{\rho_\mathcal{G}}}\right)\text{.}
  \end{equation}
\end{definition}

Note that $\mathbf T_\mathcal{G}$ is not just a normalized version of $\mathbf L$, but a completely 
new operator built from $\mathbf L$. A similar operator can be defined using the normalized 
Laplacian matrix:

\begin{definition}[Graph Translation with $\mathcal{L}$]\label{def:trans_l_norm}
  \begin{equation}
    \mathcal{T}_\mathcal{G} = \exp\left(-\imath \pi \sqrt{\frac{\mathcal{L}}{2}}\right)\text{.}
  \end{equation}
\end{definition}

We refer the interested reader to \cite{Girault.THESIS.2015} for illustrations on the comparison 
between these two operators. 
A similar definition can also be devised with the adjacency matrix-based GSP framework. We denote 
it $\mathcal{A}_\mathcal{G}$ to stress the fact that this is more than just a normalization of the 
adjacency matrix.

\begin{definition}[Graph Translation with $A$]\label{def:trans_a}
  \begin{equation}
    \mathcal{A}_\mathcal{G}
    = \exp\left(-\imath \pi \left(\mathbf I - \frac{\mathbf A}{\gamma_{\text{max}}}\right)\right)
    \text{.}
  \end{equation}
\end{definition}

\autoref{def:trans_a} rescales first the adjacency matrix to have eigenvalues of $\mathbf I - 
\frac{\mathbf A}{\gamma_{\text{max}}}$ within $[0,2]$, \textit{i.e.} the associated graph 
frequencies lie within $[0,2\pi]$. The issue of the graph frequency sign does not appear here since 
the eigenvalues of the adjacency matrix are considered as linear representations of the graph 
frequencies according to \cite{Sandryhaila.TSP.2014}.

Note that the operator $\mathcal{A}_\mathcal{G}$ is not guaranteed to be isometric when the graph 
is directed since the eigenvalues of $\mathbf A$ can be complex: if $z=x+\imath y$ is a complex 
eigenvalue of $\mathbf A/\gamma_\text{max}$, we have $\exp(-\imath\pi(1-z))=e^{-\pi 
y}\exp(-\imath\pi(1-x))$ as an eigenvalue of $\mathcal{A}_\mathcal{G}$, and $y$ being non-zero, the 
operator is not isometric. As specified in \autoref{sec:background}, we focus on undirected 
graphs which do not present such a difficulty.

\section{Graph Translation Localization Property}
\label{sec:localization}

As shown in \cite{Girault.THESIS.2015}, the graph translation does not preserve the strictest 
localization preservation property in general. The main contribution of this paper is then to show 
that all three graph translations still have localization bounds in the vertex domain. 
To that end, we give polynomial approximations of these operators that verify an exponential decay 
of the error as the degree of the polynomial increases. As shown in the introduction, this 
guarantees a polynomial decay of the impulse response from the impulse vertex to its neighbors.

Since the square root in \autoref{def:trans_l} and \ref{def:trans_l_norm} introduces an additional 
degree of complexity, we begin with $\mathcal{A}_\mathcal{G}$. All bounds are shown for the graph 
translation to the power $\alpha$, where $\alpha$ plays the role of a \textit{vertex-diffusion} 
factor. The error bounds we obtain show that localization decreases as the vertex-diffusion 
increases.

\subsection{Adjacency-based Graph Translation}

\begin{figure}[t]
  \centering
  \includegraphics{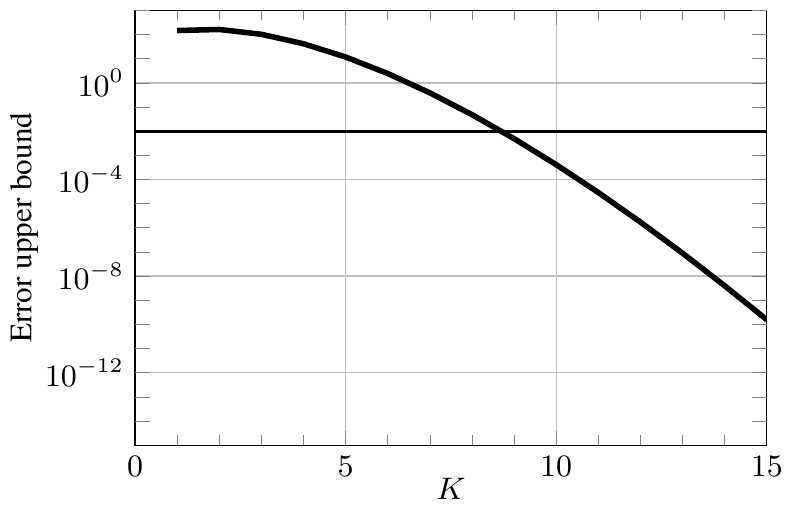}
  \vspace{-0.2cm}
  \caption{Approximation curve associated to \eqref{eq:ag_bound} with $\alpha=1$.}
  \label{fig:approx_ag}
\end{figure}

In this section, we have $\mathbf M=\mathbf I-\mathbf A/\gamma_\text{max}$, such that:
\[
  \mathcal{A}_\mathcal{G}^\alpha=\cos\left(\alpha\pi \mathbf M\right)
  -\imath\sin\left(\alpha\pi \mathbf M\right)\text{.}
\]
We consider then the following analytic formulas on $\mathbb{R}_+$:
\begin{align*}
  \cos(\alpha\pi x) &= \sum_{k=0}^{\infty} (-1)^k\frac{(\alpha\pi)^{2k}}{(2k)!}x^{2k} \\
  \sin(\alpha\pi x) &= \sum_{k=0}^{\infty} (-1)^k\frac{(\alpha\pi)^{2k+1}}{(2k+1)!}x^{2k+1}\text{,}
\end{align*}
and the following truncated sum approximation of $\mathcal{A}_\mathcal{G}^\alpha$:
\begin{equation}
  \mathcal{A}_{\mathcal{G},K}^\alpha=
  \sum_{k=0}^{K} (-1)^k\left[
    \frac{(\alpha\pi\mathbf M)^{2k}}{(2k)!}
    -\imath
    \frac{(\alpha\pi\mathbf M)^{2k+1}}{(2k+1)!}
  \right]\text{.}
\end{equation}
Using \autoref{thm:conv_op_approx_bound} and \autoref{lem:alternating_sum} with $a=0$ and 
$X=[0,1]$, we have then:
\begin{multline}
\label{eq:ag_bound}
  \frac{\left\|\left[
    \mathcal{A}_\mathcal{G}^\alpha-\mathcal{A}_{\mathcal{G},K}^\alpha
  \right]\mathbf x\right\|_2}{\|\mathbf x\|_2}
  \leq \frac{(\alpha2\pi)^{2K+2}}{(2K+2)!}
  \Scale[1.2]{\left(1+\frac{\alpha2\pi}{2K+3}\right)}
\end{multline}

\subsection{Laplacian-based Graph Translation}

\setlength{\abovedisplayskip}{6pt}
\setlength{\belowdisplayskip}{6pt}

\begin{figure}[t]
  \centering
  \includegraphics{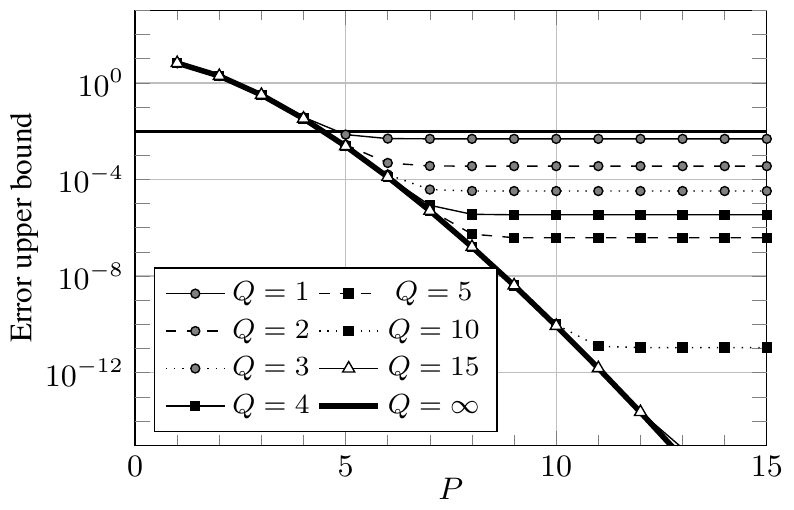}
  \vspace{-0.2cm}
  \caption{Approximation curves associated to \eqref{eq:tg_bound} for different values of $Q$ 
    with $\alpha=1$ and $\varrho=0.1$}
  \label{fig:approx_tg}
\end{figure}

We now have $\mathbf M=\mathbf L/\rho_\mathcal{G}$, such that:
\begin{equation}\label{eq:gr_trans_cos_sin}
  \mathbf T_\mathcal{G}^\alpha=
           \cos(\alpha\pi\sqrt{\mathbf M}) 
    -\imath\sqrt{\mathbf M}\frac{\sin(\alpha\pi\sqrt{\mathbf M})}{\sqrt{\mathbf M}}
\end{equation}
Let $\varrho=\lambda_1/\rho_\mathcal{G}$ be the spectral gap of $\mathbf M$ and 
$\epsilon=(1-\varrho)/(1+\varrho)<1$. Note that the case $\epsilon=1$ corresponds to $\varrho=0$, 
\textit{i.e.} the graph is disconnected. However graphs are always assumed connected for GSP since 
we can perform GSP separately on each connected component for the same results. We consider then 
the following sums:
\def\scaleRfactor{\Scale[1.1]{\frac{1}{1+\epsilon}}}
\begin{align*}
  C^{(K)}(x) &= \sum_{k=0}^{K} (-1)^k\frac{(\alpha\pi)^{2k}}{(2k)!}x^k \\
  S^{(K)}(x) &= \sum_{k=0}^{K} (-1)^k\frac{(\alpha\pi)^{2k+1}}{(2k+1)!}x^k \\
  R^{(K)}(x) &= \sqrt{\scaleRfactor}\sum_{k=0}^{K} 
    \frac{(-1)^k(2k)!}{(1-2k)(k!)^24^k}\bigl((1+\epsilon)x-1\bigr)^k\text{,}
\end{align*}
such that $\cos(\alpha\pi\sqrt{x})=C^{(\infty)}(x)$, 
$\sin(\alpha\pi\sqrt{x})/\sqrt{x}=S^{(\infty)}(x)$ for $x\in\mathbb{R}_+$, and 
$\sqrt{x}=R^{(\infty)}(x)$ for $x\in[\varrho,1]$. This last sum follows from the Taylor expansion 
of $\sqrt{1+y}$ about $y=0$ for $y\in[-\epsilon,\epsilon]$ and $y=(1+\epsilon) x-1$. We use the 
following approximation sum of the graph translation:
\begin{equation}
  T_{\mathcal{G},P,Q}^\alpha = C^{(P)}(\mathbf M) - \imath R^{(Q)}(\mathbf M) 
    S^{(P)}(\mathbf M)\text{,}
\end{equation}
Where $P$ may be different than $Q$ to account for the slower convergence speed of the sum in 
$R^{(K)}$. Indeed, Stirling's formula for the factorial shows that the leading fraction in the sum 
in $R^{(K)}$ is equivalent to $(-1)^k(\sqrt{\pi k}(1-2k))^{-1}$, such that the convergence speed of 
$S^{(K)}$ and $C^{(K)}$ are much faster than that of $R^{(K)}$. We have then the following bounds:
\begin{align*}
  \kappa_C(P) &= \frac{(\alpha\pi)^{2P+2}}{(2P+2)!} &
  \kappa_S(P) &= \frac{(\alpha\pi)^{2P+3}}{(2P+3)!}
\end{align*}
\def\scalerfactor{\Scale[1]{\frac{1}{1-\epsilon^2}}}
\begin{equation*}
  \kappa_R(Q-1) = \sqrt{\scalerfactor}
      \frac{(2Q)!}{(2Q-1)(Q!)^24^Q} 
      \left(\epsilon(1-\epsilon)\right)^Q\text{,}
\end{equation*}
using \autoref{lem:alternating_sum} on $[0,1]$ for $C$ and $S$, and \autoref{lem:analyt_fun} on 
$[\varrho,1]$ for $R$ with $f(y)=\sqrt{1+y}$ on $[-\epsilon,\epsilon]$. Notice also that 
$S^{(P)}(0)=0=S(0)$ for all $P$, such that:
\begin{align*}
  R^{(Q)}(\mathbf M)&S^{(P)}(\mathbf M)\mathbf x \\
  &= \sum_{l=0}^{N-1}
    R^{(Q)}\left(\frac{\lambda_l}{\rho_\mathcal{G}}\right)
    S^{(P)}\left(\frac{\lambda_l}{\rho_\mathcal{G}}\right)
    \widehat{x}(l) \mathbf u^{(l)} \\
  &= \sum_{l=1}^{N-1}
    R^{(Q)}\left(\frac{\lambda_l}{\rho_\mathcal{G}}\right)
    S^{(P)}\left(\frac{\lambda_l}{\rho_\mathcal{G}}\right)
    \widehat{x}(l) \mathbf u^{(l)}\text{,}
\end{align*}
for all $P$ and $Q$, with $\mathbf u^{(l)}$ the $l^\text{th}$ Fourier mode. It is then enough to 
have $\kappa_R$ and $\kappa_{RS}$ for the interval $[\lambda_1/\rho_\mathcal{G},1]=[\varrho,1]$ to 
obtain the result of \autoref{thm:conv_op_approx_bound}. This also shows:
\[
  \mathbf T_\mathcal{G}=\mathbf T_{\mathcal{G},\infty,\infty}\text{.}
\]
Using \autoref{lem:graph_op_comp}, we obtain on $[\varrho,1]$:
\[
  \kappa_{RS}(P,Q)=\kappa_R(Q)+\kappa_S(P)\left(1+\kappa_R(Q)\right).
\]
and finally:
\begin{multline}
\label{eq:tg_bound}
  \left\|\left[\mathbf T_\mathcal{G}^\alpha- \mathbf T_{\mathcal{G},P,Q}^\alpha\right]\mathbf 
    x\right\|_2
  \\
  \leq 
  \left[
  \kappa_C(P)+\kappa_S(P)
  +\kappa_R(Q)(1+2\kappa_S(P))
  \right]\left\|\mathbf x\right\|_2\text{.}
\end{multline}

Overall, this bound is dominated by the term $\kappa_R(Q)$ since $\kappa_C(P)$ and $\kappa_S(P)$ 
decrease very quickly with $P$ due to the factorial. \autoref{fig:approx_tg} shows the evolution of 
$\|[\mathbf T_\mathcal{G}-\mathbf T_{\mathcal{G},P,Q}]\mathbf x\|_2/\|\mathbf x\|_2$ according to 
$P$ and $Q$. We see that the error is mainly explained by $Q$: Larger values of $P$ do not 
decrease the error, with a plateau on the error explained by $Q$. The approximation error is 
about 1\% when $P=5$ and $Q=1$, and each increment of one of $Q$ leads to a an approximation error 
divided by about $1/\epsilon$ (for a large $P$). With these values of $P$ and $Q$, we have 
approximated $\mathbf T_\mathcal{G}$ with a local operator acting in a $6$-hops neighborhood of a 
given vertex.

Note that the faster convergence compared to \autoref{fig:approx_ag} is explained by two factors. 
First, the graph frequencies of $\mathbf L$ lie in $[0,1/2]$ while those of $\mathbf A$ lie in 
$[0,1]$ such that the error curve is slightly shifted towards bigger values of $K$ on 
\autoref{fig:approx_ag}. Second, the use of the square root of $\mathbf L$ for $\mathbf 
L_\mathcal{G}$ compared to the plain matrix $\mathbf A$ for $\mathcal{A}_\mathcal{G}$ leads to a 
steeper slope on \autoref{fig:approx_tg} (for $Q=\infty$), and a better approximation.

\begin{figure}[t]
  \centering
  \includegraphics{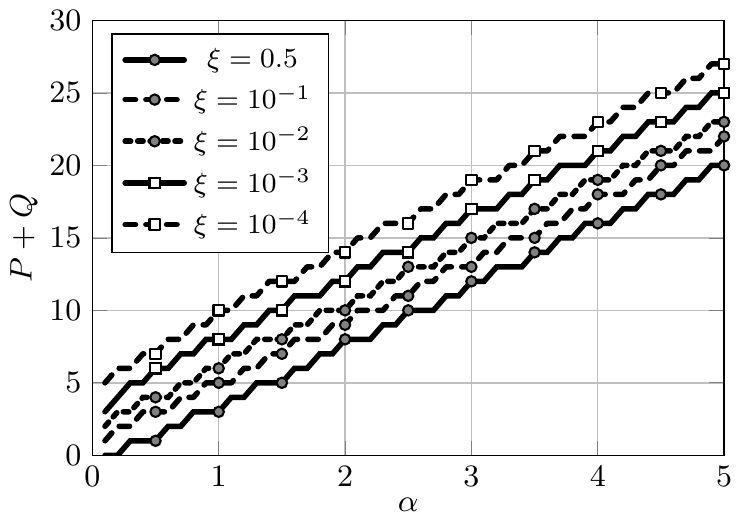}
  \vspace{-0.3cm}
  \caption{Minimal value of $P+Q$ to have a maximum error of $\xi$ in \eqref{eq:tg_bound} for 
  $\varrho=0.1$.}
  \label{fig:optim_approx}
  \vspace{-0.5cm}
\end{figure}

\autoref{fig:optim_approx} shows the minimal value of $P+Q$ such that the approximation in 
\eqref{eq:tg_bound} yields an error smaller than $\xi$ from $0.5$ to $10^{-4}$ and 
for different values of the vertex-diffusion factor $\alpha$. We observe that the more the graph 
translation is applied, the more diffused the signal can be with a looser bound. For $\alpha=1$, we 
obtain the result that 50\% of the energy is within a 3-hops radius, 90\% of the energy is within a 
5-hops radius and 99\% is within a 6-hops radius. Approximating the graph translation through a 
polynomial operator is therefore a trade-off between accuracy and the size of the neighborhood. 
Also, the slopes on \autoref{fig:optim_approx} are steep and the localization decreases much faster 
as $\alpha$ increases compared to the graph shift operator (having a unitary slope), illustrating 
another trade-off between localization and isometry.

Finally, note that this is an upper bound, and the output may remain highly localized in the vertex 
domain no matter the value of $\alpha$. This is especially the case if one Fourier mode is highly 
localized: If $\mathbf u^{(l)}\approx \mathbf \delta_i$, then $\mathbf T_\mathcal{G}^\alpha\delta_i 
\approx \mathbf T_\mathcal{G}^\alpha u^{(l)}=e^{-\imath\alpha\nu_l}\mathbf u^{(l)}\approx 
e^{-\imath\alpha\nu_l}\delta_i$, \textit{i.e.} $\mathbf T_\mathcal{G}^\alpha\delta_i$ is localized 
about vertex $i$. Our bound corresponds then to the worst case scenario of delocalized Fourier 
modes.

\subsection{Normalized-Laplacian-based Graph Translation}

The technique is the same as before using $\mathbf{M}=\mathcal{L}/2$ instead of $\mathbf{M}=\mathbf 
L/\rho_\mathcal{G}$, and leading to the exact same bound with $\varrho=\mu_1/2$:
\begin{multline}
\label{eq:tg_norm_bound}
  \left\|\left[\mathcal T_\mathcal{G}^\alpha- \mathcal T_{\mathcal{G},P,Q}^\alpha\right]\mathbf 
    x\right\|_2
  \\
  \leq 
  \left[
  \kappa_C(P)+\kappa_S(P)
  +\kappa_R(Q)(1+2\kappa_S(P))
  \right]\left\|\mathbf x\right\|_2\text{.}
\end{multline}

\section{Conclusion}

We have shown in this paper that the graph translation defined in \cite{Girault.SPL.2015} indeed 
preserves the localization of a graph signal in the sense of an exponential decay, as suggested by 
the numerical experiments. Equivalently, we showed that this operator is acting as a diffusion 
operator in the sense that the energy spreads from one vertex to its neighbors. Also, when 
iterated, the bounds are less and less tight such that the energy may spread across farther and 
farther vertices of the graph just like a diffusion process would. Note that the bounds we showed 
only depend on the spectral gap of the Laplacian, and as such are very general bounds. We refer the 
interested reader to \cite{Abreu.LAA.2007} for a review on the spectral gap of the Laplacian 
(called the algebraic connectivity of the graph). In addition, we showed several techniques 
applicable to other operators that can be used to give them polynomial approximation and to bound 
the error and the spread of the impulse response.

One last noticeable remark is that the impulse response of the graph translation can be seen as a 
generalized translation of a particular graph signal: $T_\mathcal{G}\delta_i=h*\delta_i$ where 
$\widehat{h}(l)= \exp(-\imath 2\pi\nu_l)$. This property allows to further study the impulse 
response of the graph translation using the results dedicated to the generalized translation, and 
will be used in a future paper.

\bibliographystyle{IEEEbib}
\bibliography{bibliography}

\end{document}